\documentclass[11pt]{amsart}
\usepackage{a4wide}
\usepackage{cite}
\usepackage{amsmath, amsfonts, amssymb,graphicx}
\usepackage{mathrsfs,url}
\usepackage{tkz-graph}
\DeclareMathOperator{\Csp}{CSP}

\newcommand{\cproblem}[3]{
\vspace{.2cm}
\noindent {\bf #1} \\
INSTANCE: #2 \\
QUESTION: #3 \\}

\DeclareMathOperator{\Aut}{Aut}
\DeclareMathOperator{\End}{End}
\DeclareMathOperator{\Sym}{Sym}

\author{Manuel Bodirsky}
    \address{Laboratoire d'Informatique  (LIX), CNRS UMR 7161\\
    \'{E}cole Polytechnique \\91128 Palaiseau\\
    France}
    \email{bodirsky@lix.polytechnique.fr}
    \urladdr{http://www.lix.polytechnique.fr/~bodirsky/}
\thanks{The research leading to these results has received funding from the European Research Council under the European Community's Seventh Framework Programme (FP7/2007-2013 Grant Agreement no.\ 257039) and from the EPSRC grant EP/H00677X/1. The third author was supported by the LIX-Qualcomm Postdoctoral Fellowship.}

\author{H.\ Dugald Macpherson}
    \address{School of Mathematics, University of Leeds\\
      Leeds LS2 9JT\\
      England}
    \email{h.d.macpherson@leeds.ac.uk}

\author{Johan Thapper}
    \address{Laboratoire d'Informatique  (LIX), CNRS UMR 7161\\
    \'{E}cole Polytechnique \\91128 Palaiseau\\
    France}
    \email{thapper@lix.polytechnique.fr}
    \urladdr{http://www.lix.polytechnique.fr/~thapper/}

\title[Semi-lattice Polymorphisms]{Constraint Satisfaction Tractability from Semi-lattice Operations on Infinite Sets}

\newtheorem{theorem}{Theorem}[section]
\newtheorem{lemma}[theorem]{Lemma}
\newtheorem{proposition}[theorem]{Proposition}
\newtheorem{conjecture}[theorem]{Conjecture}
\newtheorem{question}{Question}[section]
\newtheorem{definition}{Definition}[section]
\theoremstyle{remark}
\newtheorem{example}{Example}
\newcommand{\ignore}[1]{}

\begin{document}

\tikzstyle{vertex8}=[draw,circle,fill=black,text=white,minimum size=2pt,inner sep=0pt]
\tikzstyle{edge}=[-latex]
\tikzstyle{edged}=[-latex,dashed]

\begin{abstract}
  A famous result by Jeavons, Cohen, and Gyssens shows that every constraint satisfaction problem (CSP) 
  where the constraints are preserved by
  a \emph{semi-lattice operation} can be solved in polynomial time. This is one of the basic facts for the so-called
  \emph{universal-algebraic approach} to a systematic theory
  of tractability and hardness in finite domain constraint satisfaction.
  
  Not surprisingly, the theorem of Jeavons et al.\ fails for arbitrary infinite domain CSPs. 
  Many CSPs of practical interest, though, and in particular
  those CSPs that are motivated by \emph{qualitative reasoning calculi} from Artificial Intelligence, 
  can be formulated with constraint languages that are rather well-behaved from a model-theoretic point of view. 
  In particular, the automorphism group of these constraint languages tends to be \emph{large} in the sense that 
  the number of orbits of $n$-subsets of the automorphism group is bounded by some function in $n$.
  
  In this paper we present a generalization of the theorem by Jeavons et al.\ to infinite
  domain CSPs where the number of orbits of $n$-subsets grows sub-exponentially in $n$, and prove that
  preservation under a semi-lattice operation for such CSPs implies polynomial-time tractability. 
  Unlike the result of Jeavons et al., this includes many CSPs that cannot be solved by Datalog.
\end{abstract}

\maketitle



\section{Introduction}
Constraint satisfaction problems are fundamental computational problems that arise in many areas of theoretical computer science. 
In recent years, a considerable amount of research has been concentrated on the classification of those constraint satisfaction problems that
can be solved in polynomial time, and those that are computationally hard. In this paper, we contribute to this line of research
and generalize an important tractability condition from finite domain constraint satisfaction to a broad class of infinite domain
constraint satisfaction problems. 

We work with the following definition of constraint satisfaction problems (CSPs),  which is well-adapted to 
treat CSPs over infinite as well as finite domains. The definition is based on the concept of a \emph{homomorphism} between
relational structures, and equivalent to the standard definition for finite domain CSPs. 
A \emph{(relational) structure} $A$ consists of a (not necessarily finite) \emph{domain} $D(A)$
(or simply $A$ when no confusion can arise),
and a set of relations on $D(A)$, each of a finite positive arity.
Each relation is named by a \emph{relation symbol} $R$;
the corresponding relation in $A$ is denoted by $R^A$.
The set of all relation symbols is called the \emph{signature} of the structure.
A \emph{homomorphism} from a relational structure $A$ to a relational structure $B$ over the same signature is a mapping $f : D(A) \rightarrow D(B)$ such that
for each relation symbol $R$, and tuple $t \in R^A$, it holds that
$f(t) \in R^B$, where $f$ is applied component-wise to $t$.
The existence of a homomorphism from $A$ to $B$ is denoted by $A \rightarrow B$.
For a fixed structure, traditionally denoted by $\Gamma$, with finite relational signature $\tau$ 
the \emph{constraint satisfaction problem for $\Gamma$} (denoted by $\Csp(\Gamma)$) is
the following problem.

\cproblem{CSP($\Gamma$)}
{A finite structure $A$ over the signature $\tau$.}
{Is there a homomorphism from $A$ to $\Gamma$?}

To give an example, the three-colorability problem can be formulated as $\Csp(K_3)$, where
$K_3$ is the complete graph with three elements. A basic example of an infinite-domain CSP is $\Csp(({\mathbb Q}; <))$, where $({\mathbb Q};<)$ is the strict linear ordering of the 
rationals.

Jeavons and co-authors~\cite{JeavonsClosure,JBK} made the ground-breaking observation that for finite structures $\Gamma$, the complexity of $\Csp(\Gamma)$ is captured by the \emph{polymorphisms} of $\Gamma$, defined as follows. 
When $f : D^k \rightarrow D$ is a $k$-ary function, and $R$ is an an $n$-ary relation over $D$, then we say that
\emph{$f$ preserves $R$} if  for all $n$-tuples
$t_1, \dots, t_k \in R$, we have
 \[
 (f(t_1[1],\dots,t_k[1]), \dots, f(t_1[n],\dots,t_k[n])) \in R.
 \]
 A \emph{polymorphism} of a relational structure $\Gamma$ with domain $D=D(\Gamma)$ is a function from $D^k$ to $D$ that preserves all relations of $\Gamma$. In other words, a polymorphism
 is a homomorphism from $\Gamma^k$ to $\Gamma$, for some $k$.
The exploitation of polymorphisms for classifying the complexity of CSPs is sometimes referred to as the \emph{universal-algebraic approach}.
Indeed, very often tractability of $\Csp(\Gamma)$ is linked
to polymorphisms of $\Gamma$ with certain `good properties'. 

For a finite domain $D = D(\Gamma)$, it is known that $\Csp(\Gamma)$ is NP-hard unless $\Gamma$ has a polymorphism $f : D^k \rightarrow D$ satisfying
\[
f(x_1,x_2,\dots,x_k) = f(x_2,x_3,\dots,x_k,x_1),
\]
for all $x_1,\dots,x_k \in D$,
and it is conjectured that $\Csp(\Gamma)$ can be solved in polynomial-time
whenever $\Gamma$ has such a polymorphism~\cite{BartoKozikLICS10}\footnote{Barto and Kozik~\cite{BartoKozikLICS10} call $f$ a \emph{cyclic term} when it
satisfies the additional requirement of being idempotent; $f(x,\dots,x) = x$ for all $x \in D$.
They state the conditions for NP-hardness and the conjecture for polynomial-time tractability in terms of the absence or presence of such cyclic terms among the polymorphisms of the \emph{core of $\Gamma$}, cf.~Section~\ref{sect:mc-cores}. It is not hard to verify that their condition is equivalent to the one given here.}.
This conjecture is known to hold in several special cases when $f$ satisfies stronger identities.
We will now look at one such case.

An operation $f: D^2 \rightarrow D$ is called 
\begin{itemize}
\item \emph{idempotent} if $f(x,x)=x$;
\item \emph{commutative} if $f(x,y)=f(y,x)$
for all $x,y \in D$;
\item \emph{associative} if $f(x,f(y,z))=f(f(x,y),z)$ for all $x,y,z \in D$;
\item a \emph{semi-lattice operation} if it is commutative, associative, and idempotent.
\end{itemize}

Jeavons, Cohen, and Gyssens~\cite{JeavonsClosure} proved that for every finite structure $\Gamma$ with a polymorphism that
is a semi-lattice operation, the problem $\Csp(\Gamma)$ can be solved in polynomial time. 
In this paper, we present a generalization of this result to a large class of infinite-domain CSPs. 
All infinite structures considered are assumed to be countably infinite.
To state the result, we need the following definitions.

A bijective homomorphism with an inverse that is also a homomorphism is called an \emph{isomorphism}.
An \emph{automorphism} of a relational structure $\Gamma$ is an isomorphism between $\Gamma$ and itself, and the set of all automorphisms of $\Gamma$ is denoted by $\text{Aut}(\Gamma)$. 
For a subset $S$ of the domain of $\Gamma$, the \emph{orbit of $S$} in $\Gamma$ is the
set $\big \{ \{ \alpha(s) \mid s \in S\} \mid \alpha \in \text{Aut}(\Gamma) \big \}$. 
When $S$ is of cardinality $n$, then
we call the orbit of $S$ in $\Gamma$ an \emph{orbit of $n$-subsets}.
If the number of orbits of $n$-subsets of $\Gamma$ is at least $c^n$ for some $c > 1$ and all sufficiently large $n$, then we say that $\Gamma$ has \emph{exponential growth}.
Otherwise, we say that $\Gamma$ has sub-exponential growth, or that $\Gamma$ is a \emph{sub-exponential structure}.
Note that every finite structure is a sub-exponential structure
since for $n$ greater than the domain size, 
there are no $n$-subsets at all, hence zero orbits.
But also the structure $({\mathbb Q};<)$ (and all structures with domain ${\mathbb Q}$ whose relations are first-order definable in $({\mathbb Q};<)$) is sub-exponential: it has only one orbit of $n$-subsets,
for all $n$ (see e.g.~\cite{Hodges}). 
Our main result is the following.

\begin{theorem}\label{thm:main}
Let $\Gamma$ be a sub-exponential structure with finite relational signature. If $\Gamma$ has a semi-lattice polymorphism, then $\Csp(\Gamma)$ can be be solved in polynomial time.
\end{theorem}

Finite domain structures with a semi-lattice polymorphism can be solved by a standard technique, known as \emph{establishing arc-consistency}~\cite{DechterBook}.
The situation is different for sub-exponential structures with a semi-lattice polymorphism.
Consider for instance the structure $({\mathbb Q}; \{(x,y,z) \mid x>y \vee x>z \})$.
It has the same automorphism group as $({\mathbb Q};<)$ and hence is sub-exponential as well.
This structure has the function $(x,y) \mapsto \min(x,y)$ on ${\mathbb Q}$ as a polymorphism,
but it has been shown that there is no $k$ such that this problem
can be solved by establishing $k$-consistency~\cite{ll}. 
Tractability of $\Csp(({\mathbb Q}; \{(x,y,z) \mid x>y \vee x>z \}))$ 
instead follows from the following more general result of~\cite{tcsps-journal}: 
every structure $\Gamma$ with domain $\mathbb Q$ that has the same automorphism
group as $({\mathbb Q};<)$ and that is preserved by the minimum function
(or the maximum function) has a polynomial-time tractable CSP.
Our result will be a proper generalization of this result and of the mentioned result of Jeavons, Cohen, and Gyssens.

Let us remark that for a general infinite structure, a semi-lattice polymorphism does not suffice to ensure tractability.
For an arbitrary subset $U$ of ${\mathbb N}$, let $\Gamma_U$ be the
structure $({\mathbb N}; \{(x,y) \in {\mathbb N}^2 \; | \; x=y+1\}, \{0\}, U)$.
Every such structure has min (and max) as a semi-lattice polymorphism.
We claim that CSP($\Gamma_U$) and CSP($\Gamma_{V}$) are different problems
for distinct subsets $U$ and $V$ of ${\mathbb N}$.
Let $m$ be any element in $U \, \Delta \, V$, where $\Delta$ denotes the symmetric difference of the sets, and
let $A_m$ be the instance on variables $\{x_0,\dots,x_m\}$ containing the
constraints
$x_0 = 0$, $x_{i} = x_{i-1}+1$ for $1 \leq i \leq m$, and $U(x_m)$.
Then $A_m$ is a satisfiable instance for precisely one of the problems
CSP($\Gamma_U$) and CSP($\Gamma_{V}$).
It follows that there are as many pairwise distinct CSPs of this type as there are subsets of ${\mathbb N}$, i.e., uncountably many.
However, there are no more than countably many algorithms, hence
CSP$(\Gamma_U)$ is undecidable for some $U$.

The example of the previous paragraph relies on the fact that for all structures $\Gamma_U$, there is an \emph{infinite} number
of orbits of $2$-subsets. When the number of orbits of $n$-subsets of $\Gamma$ is finite for all $n$, then the structure is called \emph{$\omega$-categorical} in model theory~\cite{Hodges}. 
For $\omega$-categorical structures, polymorphisms still capture the
computational complexity of $\Csp(\Gamma)$~\cite{BodirskyNesetrilJLC}.
Our main result can thus be seen as a contribution to the further extension of
the universal-algebraic approach from finite to $\omega$-categorical
structures.

\subsection*{Overview}
Our paper is structured as follows. In Section~\ref{sect:alg} we introduce a new algorithmic technique to solve infinite-domain
constraint satisfaction problems, and present a reduction of $\Csp(\Gamma)$ for structures $\Gamma$ with
a semi-lattice polymorphism to an \emph{efficient sampling algorithm} for $\Gamma$. 
The basic idea is that when there is such an efficient sampling algorithm for $\Gamma$, then we can use the 
arc-consistency procedure for finite domain CSPs to solve $\Csp(\Gamma)$ (we actually use the \emph{uniform} version of the
arc-consistency procedure where both $A$ and a finite template $B$ are part of the input). 
In fact, this technique works under a slightly more general assumption on $\Gamma$: instead of requiring the existence
of a semi-lattice polymorphism, we only require that $\Gamma$ has \emph{totally symmetric} polymorphisms of all arities. 

The next part of our paper, Section~\ref{sect:classification}, is devoted to the proof that all sub-exponential structures $\Gamma$ with totally symmetric polymorphisms of all arities admit such an efficient sampling algorithm. Here, our proof is based on a \emph{classification} of those structures $\Gamma$. 
We would like to remark that the general algorithmic technique is applicable also for many structures with totally symmetric
polymorphisms of all arities that are \emph{not} sub-exponential, and this will be illustrated by some examples in Section~\ref{sect:beyond}. In fact, we make the conjecture that when $\Gamma$ is an $\omega$-categorical structure with a semi-lattice polymorphism,
then $\Csp(\Gamma)$ is in P. However, unlike the case of sub-exponential structures, we cannot provide a classification result like the one in Section~\ref{sect:classification} for this more general case, and so this remains an interesting open question.

\section{Algorithm}
\label{sect:alg}

One of the basic building blocks of our algorithm will be
the arc-consistency procedure.
We start this section by recalling, 
in the case of finite relational structures,
the connection between the applicability of this procedure,
homomorphisms from the \emph{set structure},
and the existence of \emph{totally symmetric polymorphisms}.

Let $B$ be a finite structure with a finite relational signature, and
let $A$ be an instance of $\Csp(B)$.
The \emph{arc-consistency procedure} (AC) applied to the problem $(A, B)$ 
works by reducing a set of possible images for each variable in $A$.
If such a set becomes empty during the procedure, it follows that there
can be no homomorphism,
so AC rejects. Otherwise, AC accepts.
An algorithm for AC is shown in Figure~\ref{fig:alg-ac}.
We say that arc-consistency \emph{solves} the problem $\Csp(B)$ if,
for every instance $A$, the procedure accepts \emph{if and only if}
$A \rightarrow B$.
Arc-consistency is sometimes called \emph{hyperarc-consistency} when applied
to structures with relations of arity greater than two.
It can be implemented to run in time that is polynomial in 
$|A|+|B|$, take e.g.\ AC-3~\cite{DechterBook}.

The \emph{set structure} of $B$, denoted by ${\mathcal P}(B)$, 
has as vertices all non-empty subsets of the domain of $B$.
For every $k$-ary relation $R^B$, we have $(U_1, \dots, U_k) \in R^{{\mathcal P}(B)}$
iff for every $i$ and $u_i \in U_i$, there exists $u_1 \in U_1, \dots, u_{i-1} \in U_{i-1}, u_{i+1} \in U_{i+1}, \dots, u_k \in U_k$ such that
$(u_1,\dots,u_k) \in R^B$.

\begin{figure}[ht]
\begin{center}
\small
\fbox{
\begin{tabular}{l}
AC($A,B$) \\
{\rm // Input: Finite relational structures $A$ and $B$ over the same signature} \\
{\rm // Accepts iff $A \rightarrow {\mathcal P}(B)$} \\
\hspace{0cm}for all $x \in D(A)$ do \\
\hspace{.5cm}$h(x) := D(B)$ \\
\hspace{0cm}repeat \\
\hspace{.5cm}for all $R^A \in A$ do \\
\hspace{1cm}for all $(x_1,\dots,x_k) \in R^A$ do \\
\hspace{1.5cm}for all $i = 1, \dots, k$ do \\
\hspace{2cm}$h(x_i) := \pi_i (R^B \cap h(x_1) \times \dots \times h(x_k))$ \\
\hspace{0cm}until $h(x)$ does not change for any $x \in D(A)$ \\
\hspace{0cm}if $h(x) = \emptyset$ for some $x \in D(A)$ then reject \\
\hspace{0cm}else accept \\
\end{tabular}}
\end{center}
\caption{Algorithm for AC($A,B$), closely following the pseudocode given in~\cite{ACandFriends}.}
\label{fig:alg-ac}
\end{figure}

A $k$-ary function $f: D^k \rightarrow D$ is called \emph{totally symmetric}
if for all $x_1, \dots, x_k, y_1, \dots, y_k \in D$ we have
\[
f(x_1,\dots,x_k) = f(y_1,\dots,y_k) \text{ whenever } \{x_1,\dots,x_k\} = \{y_1,\dots,y_k\}.
\]
We say that a structure $B$ has \emph{totally symmetric polymorphisms of all arities} if, for each $k \geq 1$, there is a $k$-ary polymorphism of $\Gamma$ that is totally symmetric.

The following is well-known, cf.~\cite{DalmauPearson,FederVardi}.

\begin{theorem}
\label{thm:equiv}
  Let $B$ be a finite structure with a finite relational signature.
  The following are equivalent.
  \begin{enumerate}
  \item
    The arc-consistency procedure solves $\Csp(B)$.
  \item
    There is a homomorphism ${\mathcal P}(B) \rightarrow B$.
  \item
    The structure $B$ has totally symmetric polymorphisms of all arities.
  \end{enumerate}
\end{theorem}


It is clear that when $\Gamma$ has a semi-lattice operation $f$, 
then it also has a totally symmetric polymorphism $f_n$ of arity $n$, 
for each $n \geq 2$:
\[
f_n(x_1,\dots,x_n) := f(x_1, f(x_2, \dots, f(x_{n-1},x_n) \dots )).
\]
From Theorem~\ref{thm:equiv}, it thus follows that arc-consistency solves $\Csp(B)$ whenever $B$ is a finite relational structure with a semi-lattice polymorphism.
In our arguments, we only need the weaker condition on 
totally symmetric polymorphisms of all arities;
this gives us a stronger result.

The other component of our algorithm will be a procedure to
efficiently ``sample'' appropriate finite substructures of an
infinite structure $\Gamma$.
Formally, we make the following definition.

\begin{definition}
  Let $\Gamma$ be a structure over a finite relational signature.
  We say that an algorithm is a \emph{sampling algorithm for $\Gamma$} if,
  given a positive integer $n$, 
  it computes a finite structure $B$ that is isomorphic to a substructure of $\Gamma$ such that $A \rightarrow B$ if and only if $A \rightarrow \Gamma$,
  for every instance $A$ with $|A| = n$.
  A sampling algorithm is called \emph{efficient} if its running time is
  bounded by a polynomial in $n$.
\end{definition}

We are now ready to give an outline of our algorithm.
Let $\Gamma$ be a sub-exponential structure over a finite relational
signature, and assume that $\Gamma$ has totally symmetric polymorphisms
of all arities.

\begin{figure}[ht]
\begin{center}
\small
\fbox{
\begin{tabular}{l}
CSP($\Gamma$) \\
{\rm // Input: A finite relational structure $A$ over the same signature as $\Gamma$} \\
{\rm // Accepts iff $A \rightarrow \Gamma$} \\
\hspace{0cm}$B$ := Sample-$\Gamma$($|A|$). \\
\hspace{0cm}if AC$(A,B)$ rejects then reject \\
\hspace{0cm}else accept \\
\end{tabular}}
\end{center}
\caption{Algorithm for $\Csp(\Gamma)$, with $\Gamma$ being a sub-exponential structure with totally symmetric polymorphisms of all arities. `Sample-$\Gamma$' is a sampling algorithm for $\Gamma$.}
\label{fig:alg-csp}
\end{figure}

The main idea of our algorithm is to reduce $\Csp(\Gamma)$ to an appropriate 
\emph{uniform} finite domain CSP.
That is, when given an instance $A$ of $\Csp(\Gamma)$, 
we reduce to the following problem:
decide whether $A$ maps homomorphically to $B$, 
where $B$ is a finite substructure of $\Gamma$ returned by
a sampling algorithm for $\Gamma$ on input $|A|$, 
and $B$ is considered as part of the input.
The algorithm is given in Figure~\ref{fig:alg-csp}.
Note that we rely on arc-consistency for deciding whether
$A$ maps homomorphically to $B$.
Hence, for this approach to work,
we need to establish the following.




\begin{enumerate}
\item There should be an efficient sampling algorithm which samples some $B$ from $\Gamma$.
\item The arc-consistency procedure applied to $(A,B)$ should accept if and only if $A \rightarrow B$.
\end{enumerate}
The first condition implies that the size of $B$ is polynomial in the size
of $A$, and since
AC can be implemented to run in time that
is polynomial in $|A|+|B|$,
it follows that our algorithm will be polynomial in $|A|$.
The second condition ensures that the algorithm gives the correct answer
for every instance $A$ of $\Csp(\Gamma)$.

It will be the purpose of Section~\ref{sect:classification} to prove
that an efficient sampling algorithm for $\Gamma$ exists.
We state this result as follows.

\begin{theorem}
\label{thm:alg-cond}
  Let $\Gamma$ be a sub-exponential structure with a finite relational signature and totally symmetric polymorphisms of all arities.
  Then there is an efficient sampling algorithm for $\Gamma$.
\end{theorem}

We next prove one part of a generalization of the equivalence between
the second and third item of Theorem~\ref{thm:equiv} to infinite domains. 
This result has a converse for all $\omega$-categorical structures,
cf.\ Section~\ref{sect:disc}.

\begin{lemma}
\label{lem:totsymtohomo}
  Let $\Gamma$ be a structure over a finite relational signature.
  If $\Gamma$ has totally symmetric polymorphisms of all arities,
  then
  ${\mathcal P}(S) \rightarrow \Gamma$ for all finite substructures
  $S \subseteq \Gamma$.
\end{lemma}

\begin{proof}
  Let $S$ be a finite substructure of $\Gamma$, and let $f$ be an $m$-ary
  totally symmetric polymorphism of $\Gamma$, where $m = k|S|$
  and $k$ is the maximum arity of any relation in $\Gamma$.
  Let $f' : {\mathcal P}(S) \rightarrow \Gamma$ be the function defined on a non-empty set
  $X = \{x_1, \dots, x_i\} \subseteq S$ by $f'(X) = f(x_1,\dots,x_i,x_i,\dots,x_i)$, where the list of arguments of $f$ is padded to length $m$ by elements already occurring in $X$.
  We claim that $f'$ is a homomorphism ${\mathcal P}(S) \rightarrow \Gamma$.
  We must show that $(f'(U_1),\dots,f'(U_k)) \in R^{\Gamma}$ for an arbitrary relation $R$, and tuples $(U_1,\dots,U_k) \in R^{{\mathcal P}(S)}$.
  For each $1 \leq i \leq k$, let $T_i \subseteq R^S \cap (U_1 \times \dots \times U_k)$ be a set of $k$-tuples such that $\{ t[i] \mid t \in T_i \} = U_i$.
  Note that $\{ t[i] \mid t \in T_j, 1 \leq j \leq k \} = U_i$,
  and that we can choose $T_i$ to have size $l(i) := |U_i| \leq |S|$,
  for each $i$.
  Let $u_{i1}, \dots, u_{il(i)}$ be an enumeration of $U_i$ and $t^{i}_{1},\dots,t^{i}_{l(i)}$ be an enumeration of $T_i$.
  We then have
  \[
  f'(U_i) = f'(\{u_{i 1}, \dots u_{i l(i)}\}) = f(t^{1}_1[i],\dots,t^{1}_{l(1)}[i],\dots,t^{k}_1[i],\dots,t^{k}_{l(k)}[i],t^{k}_{l(k)}[i], \dots, t^{k}_{l(k)}[i]),
  \]
  for all $1 \leq i \leq k$.
  Since $f$ is a polymorphism of $\Gamma$, and $t^{i}_j \in R^S \subseteq R^{\Gamma}$ for all $i$ and $j$, it follows that $(f'(U_1),\dots,f'(U_k)) \in R^{\Gamma}$.
\end{proof}

The correctness and efficiency of our algorithm now follows from the previous lemma 
in conjunction with the existence of an efficient sampling algorithm for $\Gamma$. 

\begin{theorem}
\label{thm:correctness}
  Let $\Gamma$ be a structure over a finite relational signature with
  totally symmetric polymorphisms of all arities.
  If there exists an efficient sampling algorithm for $\Gamma$,
  then the algorithm in Figure~\ref{fig:alg-csp} correctly solves $\Csp(\Gamma)$ in polynomial time.
\end{theorem}

\begin{proof}
  Let $A$ be the input structure and
  let $B$ be the structure returned by the sampling algorithm for $\Gamma$
  on input $|A|$.
  The sampling algorithm runs in polynomial time in $|A|$,
  so the size of $B$ will be polynomial in $|A|$.
  Since AC($A,B$) can be implemented to run in time polynomial
  in $|A|+|B|$, it follows that the entire algorithm runs in polynomial time.

  To show correctness, note that if AC($A,B$) rejects, 
  then $A \not\rightarrow B$ which is equivalent to $A \not\rightarrow \Gamma$ since 
  $B$ was produced by the sampling algorithm for $\Gamma$ on input $|A|$.
  We may therefore safely reject.
  Otherwise, AC$(A,B)$ terminates with a non-empty list $h(x) \subseteq B$ for each $x \in A$.
  Furthermore, for each $k$-ary relation $R^A$, tuple $(x_1, \dots, x_k) \in R^A$, index $i$, and element $d \in h(x_i)$, we know that $R^B \cap h(x_1) \times \dots \times h(x_k)$ contains a tuple $(d_1, \dots, d_k)$ with $d_i = d$.
  In other words $(h(x_1),\dots,h(x_k)) \in R^{{\mathcal P}(B)}$, so the function $h : A \rightarrow {\mathcal P}(B)$ is a homomorphism.
  By assumption, $\Gamma$ has totally symmetric polymorphisms of all arities, 
  so Lemma~\ref{lem:totsymtohomo} implies the existence of a homomorphism 
  $g : {\mathcal P}(B) \rightarrow \Gamma$.
  In conclusion, we have a homomorphism $g \circ h : A \rightarrow \Gamma$.
\end{proof}

As a direct corollary of Theorem~\ref{thm:alg-cond} and Theorem~\ref{thm:correctness}, we get the following result, which also implies our main result, Theorem~\ref{thm:main}.

\begin{theorem}
\label{thm:yyy}
  Let $\Gamma$ be a sub-exponential structure with a finite relational signature.
  If $\Gamma$ has totally symmetric polymorphisms of all arities, 
  then $\Csp(\Gamma)$ is solvable in polynomial time.
\end{theorem}

\section{Classification}
\label{sect:classification}
The purpose of this section is to prove Theorem~\ref{thm:alg-cond}, showing that every sub-exponential structure has an efficient sampling algorithm. 
Our approach is based on a \emph{classification} of sub-exponential structures with totally symmetric polymorphisms of all arities. 
The general outline of this classification is as follows. We first present an argument that reduces the classification task to
those sub-exponential structures $\Gamma$ 
that are \emph{model-complete cores} and have totally symmetric polymorphisms of all arities (Section~\ref{sect:mc-cores}).
From there on, the classification follows a decomposition of the automorphism group of $\Gamma$.
The next step is the reduction to those structures $\Gamma$ having a \emph{transitive} automorphism group (Section~\ref{ssect:general}). We then use the fact that the automorphism group of $\Gamma$ has only finitely many congruence
relations to further reduce the classification task to the case that the automorphism group of $\Gamma$ is \emph{primitive}
(Section~\ref{ssect:finest}). 
Combining the central theorem from~\cite{MacphersonOrbits} on primitive
permutation groups with a sub-exponential number of orbits of $n$-subsets
with Cameron's theorem on highly set-transitive permutation groups~\cite{Cameron5} 
we finish the classification in Section~\ref{ssect:primitive}. 

We find it instructive to give a `top-down' presentation of the classification proof, rather than starting from special cases,
and to assemble more general sub-exponential structures from specific ones.
We thus take a
decomposition approach, and show first how to describe the most general case in terms of its components.
This sometimes leads to forward-references of results, but we believe that the reader will be compensated by a more accessible presentation.

\subsection{Preliminaries}
Before we start, we recall a basic fact which will frequently be used in the following arguments, and 
which explains the interaction between permutation group theory and logic for $\omega$-categorical structures.
Let $\Gamma$ be a relational structure.
In this paper, we say that a relation $R$ over the domain of $\Gamma$ is \emph{first-order definable} in $\Gamma$ if there is a first-order formula $\phi({\bar x})$ such that $\phi({\bar a})$ is true in $(\Gamma,{\bar a})$ if and only if ${\bar a} \in R$.
We say that a relational structure $\Gamma'$ is first-order definable in $\Gamma$ if $\Gamma'$ and $\Gamma$ have the same domain and every relation $R$ in $\Gamma'$ is first-order definable in $\Gamma$.
We say that two relational structures are \emph{first-order interdefinable} if one of them is first-order definable in the other, and vice versa.

\begin{theorem}[see e.g.\ Theorem~6.3.5 in~\cite{Hodges}]\label{thm:ryll}
Let $\Gamma$ be $\omega$-categorical. Then a relation $R$ is preserved by all automorphisms of
$\Gamma$ if and only if $R$ has a first-order definition in $\Gamma$. In particular, two structures are first-order interdefinable if and only if they have the same automorphism group. 
\end{theorem}

This theorem makes possible a translation of terminology between logic and permutation groups.
We illustrate its use with the following, which will be needed later on. 
A \emph{congruence} of a permutation group 
is an equivalence relation that is preserved by all permutations in the group.
A permutation group is called \emph{primitive} if the only congruences are
the equivalence relation with just one
equivalence class, and the equivalence relation where all equivalence classes are of size one; it is called \emph{imprimitive} otherwise.
By Theorem~\ref{thm:ryll}, the congruences
of the automorphism group of an $\omega$-categorical structure $\Gamma$ are precisely the first-order definable equivalence relations of $\Gamma$. 
We will say that $\Gamma$ is \emph{primitive} if its automorphism group is primitive.

A permutation group $G$ on a countable set $X$ is called \emph{closed} if and only if it is the automorphism group of a relational 
structure with domain $X$. The topological explanations for this terminology can be found in~\cite{Gao}, Theorem 2.4.4.
The corresponding topology is called the \emph{topology of pointwise convergence} on $\Sym(X)$,
where $\Sym(X)$ denotes the set of all permutations of $X$. In this topology, the open sets are unions of sets of the form
$\{ \alpha \in \Sym(X) \mid \alpha(\bar x)=\bar y\}$, for $n$-tuples  $\bar x$, $\bar y$  of elements of $X$.
A subset $H$ of $G \subseteq \Sym(X)$ is \emph{dense (in $G$)} if the closure
of $H$ with respect to this topology equals $G$. 

\subsection{Model-Complete Cores}
\label{sect:mc-cores}
An endomorphism of a relational structure $\Gamma$ is a homomorphism from $\Gamma$ to $\Gamma$; we denote the set of all endomorphisms of $\Gamma$ by $\End(\Gamma)$. A relational structure is called a \emph{core} if every endomorphism of $\Gamma$ is an embedding\footnote{An embedding of a relational structure $\Gamma$ into a relational structure $\Delta$ is an isomorphism between $\Gamma$ and an induced substructure of $\Delta$.}.
For a relational structure $\Delta$, a \emph{core of $\Delta$} is a core structure $\Gamma$ that is homomorphically equivalent to
$\Delta$, that is, there is a homomorphism from $\Gamma$ to $\Delta$ and vice versa.
A first-order formula is called \emph{primitive positive} if it is of the form
$$ \exists x_1,\dots, x_n (\psi_1 \wedge \dots \wedge \psi_n),$$
where $\psi_1,\dots,\psi_n$ are atomic formulas. 
The importance of primitive positive definitions in this paper comes from the fact that relations with a primitive positive definition
in a relational structure $\Gamma$ are preserved by the polymorphisms of $\Gamma$.

The motivation of these definitions for constraint satisfaction with finite templates $\Delta$ comes from the following facts.
\begin{itemize}
\item Every finite relational structure $\Delta$ has a core $\Gamma$, and $\Gamma$ is unique up to isomorphism.
\item When $\Gamma$ is a core of $\Delta$, then $\Gamma$ and $\Delta$ have the same CSP. 
\item In a finite core structure $\Gamma$, every orbit of $n$-tuples\footnote{When $(t_1,\dots,t_n)$ is an $n$-tuple of elements of $\Gamma$, then the \emph{orbit of $t$ (in $\Gamma$)} is the set $\{ (\alpha(t_1),\dots,\alpha(t_n)) \; | \; \alpha \in \text{Aut}(\Gamma) \}$.} is primitive positive definable. 
\end{itemize}

These properties have been generalized to 
$\omega$-categorical structures.
To precisely state the result, we need the following concepts. 
A structure $\Gamma$ is \emph{model-complete} if every embedding of $\Gamma$ into itself preserves all first-order formulas.
We later need the following.

\begin{lemma}[Combination of Theorem 18 in~\cite{Cores-journal} and proof of Corollary 7 in~\cite{BodirskyPinskerRandom}]\label{lem:mc-cores}
An $\omega$-categorical relational structure $\Delta$ is a model-complete core if and only if the group of automorphisms of $\Delta$ is dense in the endomorphism monoid of $\Delta$, that is, for every endomorphism $f$ and finite subset $U$ of $D(\Delta)$, there is an automorphism $g$ of $\Delta$ agreeing with $f$ on $U$.
\end{lemma}

\begin{theorem}[Theorem 16 in~\cite{Cores-journal}; see also~\cite{BodHilsMartin-Journal}]
\label{thm:mc-cores} Let $\Delta$ be an $\omega$-categorical relational structure. Then
\begin{enumerate}
\item\label{enum:hascore}
$\Delta$ is homomorphically equivalent to a model-complete core $\Gamma$;
\item\label{enum:unique} 
the structure $\Gamma$ is unique up to isomorphism, and $\omega$-categorical or finite;
\item\label{enum:ppdef} 
in $\Gamma$, every orbit of $n$-tuples is primitive positive definable. 
\end{enumerate}
\end{theorem}

For our classification project (and our algorithmic result), it therefore suffices to study the CSPs for model-complete cores of sub-exponential structures. 
Let us first show that the model-complete core of a sub-exponential structure is again sub-exponential. This follows from the following
more general result.

\begin{proposition}\label{prop:orbit-non-growth}
Let $\Delta$ be an $\omega$-categorical relational structure, and let $\Gamma$ be its model-complete core. Then for every $n$,
the number of orbits of $n$-subsets in $\Gamma$ is at most the number of orbits of $n$-subsets in $\Delta$.
\end{proposition}
\begin{proof}
Let $f$ be a homomorphism from $\Delta$ to $\Gamma$, and $g$ be a homomorphism from $\Gamma$ to $\Delta$.
Since $\Gamma$ is a core, it follows that $f \circ g$ is an embedding, so $g$ is injective.
It now suffices to show that when two $n$-subsets $t_1, t_2$ of $\Gamma$
are mapped by $g$ to two $n$-subsets $s_1, s_2$ in the same orbit of
$n$-subsets in $\Delta$, 
then $t_1$ and $t_2$ lie in the same orbit of $n$-subsets in $\Gamma$.
Let ${\bar t_1}$ be an $n$-tuple listing all the elements in $t_1$,
let $\alpha$ be an automorphism of $\Delta$ that maps $s_1$ to $s_2$, and
let ${\bar s_2} = \alpha(g({\bar t_1}))$.
Since ${\bar s_2}$ lists all the elements of $s_2$, 
we can arrange the elements of $t_2$ into an $n$-tuple ${\bar t_2}$
such that $g({\bar t_2}) = {\bar s_2}$.
By Theorem~\ref{thm:mc-cores}(\ref{enum:ppdef}), 
there are primitive positive definitions $\phi_1$ and $\phi_2$ of the orbits of ${\bar t_1}$ and ${\bar t_2}$. 
Since $g$, $\alpha$, and $f$ preserve primitive positive formulas, the tuple ${\bar t_3} := f(\alpha(g({\bar t_1})))$ satisfies $\phi_1$.
But $f(\alpha(g({\bar t_1}))) = f(g({\bar t_2}))$, and hence ${\bar t_3}$ also satisfies $\phi_2$.
Therefore, $\phi_1$ and $\phi_2$ define the same orbit of $n$-tuples, and so ${\bar t_1}$ and ${\bar t_2}$ are in the same orbit.
This implies that $t_1$ and $t_2$ are in the same orbit of $n$-subsets.
\end{proof}


In general it might not be true that the model-complete core of a sub-exponential structure with a semi-lattice polymorphism has again a semi-lattice polymorphism. A finite example of this situation can be derived from Proposition~5.2 in~\cite{LZFinitePosets}. This example shows a finite poset with a semi-lattice polymorphism which retracts to a poset without a semi-lattice polymorphism. By introducing constants, the latter structure can be turned into a core of the former.
However, we always have the following.

\begin{proposition}
Let $\Delta$ be an $\omega$-categorical relational structure with an $n$-ary  totally symmetric polymorphism. 
Then the model-complete core of $\Delta$ also has an $n$-ary totally symmetric polymorphism.
\end{proposition}
\begin{proof}
Let $\Gamma$ be the model-complete core of $\Delta$, and
let $g : \Delta \rightarrow \Gamma$ and $h : \Gamma \rightarrow \Delta$
be homomorphisms.
When $f$ is an $n$-ary totally symmetric polymorphism of $\Delta$, 
then $f': D(\Gamma)^n \rightarrow D(\Gamma)$ 
defined by $(x_1,\dots,x_n) \mapsto g(f(h(x_1),\dots,h(x_n)))$
is totally symmetric, and a polymorphism of $\Gamma$.
\end{proof}

Once we have a sampling algorithm for the core $\Gamma$ of a structure $\Delta$, we obtain a sampling algorithm for $\Delta$ as follows.
Let $h : \Gamma \rightarrow \Delta$ be a homomorphism
and let $B'$ be a sample returned by the sampling algorithm for $\Gamma$ on input $n$.
We then let the sample $B$ of $\Delta$ be the substructure induced by the image of $B'$ under $h$.
For any structure $A$ of size $n$ it follows that $A \rightarrow \Delta$
implies $A \rightarrow \Gamma$ which implies $A \rightarrow B' \rightarrow B$, so $B$ is indeed a sample of $\Delta$.
The structure $B$ can be computed in polynomial time in the size of $B'$
which in turn is polynomial in $n$, so
$B$ can be computed in polynomial time in $n$.
We conclude that it suffices to show Theorem~\ref{thm:alg-cond} for the special case of sub-exponential structures that are model-complete cores 
and have totally symmetric polymorphisms of all arities. 
\subsection{Reduction to the Transitive Case}
\label{ssect:general}
Let $\Gamma$ be a sub-exponential model-complete core. 
Since $\Gamma$ is sub-exponential, it has in particular a finite number of orbits of 1-subsets, called \emph{orbits} for short.
A structure is called \emph{transitive} if it has only one orbit.

\begin{proposition}\label{prop:trans-mc-cores}
Let $\Gamma$ be an $\omega$-categorical 
model-complete core, and 
let $\Gamma'$ be the expansion of $\Gamma$ by all
primitive positive definable relations. 
Let $U$ be an orbit of $\Gamma$, and $\Delta'$ be 
the restriction of $\Gamma'$ to $U$. Then 
$\Delta'$ is a transitive model-complete core. 
\end{proposition}
\begin{proof}
First observe that every automorphism $\alpha$ of $\Gamma$
is also an automorphism of $\Gamma'$ and preserves $U$,
and hence $\alpha|_U$ is an endomorphism
of $\Delta'$. Since the same also applies to the inverse of $\alpha$, we have that also $\alpha|_U$ has an inverse in $\End(\Delta')$, and therefore is an automorphism of $\Delta'$. 
So, the restriction of an automorphism of $\Gamma$ to $U$ is an automorphism of $\Delta'$, and therefore $\Delta'$ is transitive.
 
To show that $\Delta'$ is a model-complete core, 
let $e$ be an endomorphism of $\Delta'$, and
let $t$ be a $k$-tuple of elements from $U$. Then any
primitive positive formula that holds on $t$ in $\Gamma$
also holds on $e(t)$, since $\Delta'$ is the restriction
of an expansion of $\Gamma$ by all primitive positive
definable relations. Since $\Gamma$ is a model-complete core, the orbits of $k$-tuples are primitive positive definable in $\Gamma$ by Theorem~\ref{thm:mc-cores}(3), and hence there is an automorphism $\alpha$ of 
$\Gamma$ that maps $e$ to $e(t)$. Then $\alpha|_U$
is an automorphism of $\Delta'$. This shows that $\Aut(\Delta')$ is dense in $\End(\Delta')$, and the statement 
follows from Lemma~\ref{lem:mc-cores}.
\end{proof}

The following result is proved in Section~\ref{ssect:finest} and Section~\ref{ssect:primitive}. 

\begin{theorem}\label{thm:transitive}
Let $\Delta$ be a transitive sub-exponential model-complete core with totally symmetric polymorphisms of all arities.
Then $\Delta$ is either a structure of size 1, or it is isomorphic to a structure which is first-order interdefinable with $({\mathbb Q};<)$.
Any such structure
with a finite relational signature has an efficient sampling procedure. 
\end{theorem}

So we are left with the task to design an efficient sampling algorithm
for $\Gamma$ using the efficient sampling algorithms for each
substructure of $\Gamma$ induced by an orbit of $\Gamma$. For this, we
need to analyze how the automorphism group $G$ of $\Gamma$ is built from its \emph{transitive constituents},
that is, from 
the permutation groups of the form $\big \{ \alpha|_{U} \mid \alpha \in \Aut(\Gamma) \big \}$ on $U$ where $U$ is an orbit of $\Gamma$. 
In general, we only know that $G$ is a subdirect product of its transitive constituents (see, e.g.,~\cite{CameronPermutationGroups}).
In our case, we can make this decomposition more precise, 
since we have a good knowledge of the group $\Aut(({\mathbb Q};<))$.

\begin{lemma}\label{lem:convex}
\label{lem:convexorbits}
  Let $\Gamma$ be a sub-exponential structure with
  an $\Aut(\Gamma)$-invariant linear order $<$ defined on
  the union of two orbits $U$ and $V$.
  Then $U$ and $V$ are convex with respect to $<$.
\end{lemma}

\begin{proof}
  Assume to the contrary that there are elements $u_1, u_2 \in U$, $v_1 \in V$ such that $u_1 < v_1 < u_2$.
  Since $u_1$ and $u_2$ lie in the same orbit, there is an automorphism $\alpha$ of $\Gamma$ such that $\alpha(u_1) = u_2$.
  Let $v_{i+1} = \alpha(v_i)$ for $i = 1, \dots, m-1$, and let
  $u_{i+1} = \alpha(u_i)$ for $i = 2, \dots, m-1$.
  Since $\alpha$ preserves the order $<$, we have $u_2 = \alpha(u_1) < \alpha(v_1) = v_2$, and $v_2 < \alpha(u_2) = u_3$.
  By repeated application of $\alpha$, we obtain $u_i < v_i < u_{i+1}$ for all $i = 1, \dots, m-1$, and finally $u_m < v_m$.
  Hence, we can encode sequences $s \in \{0,1\}^m$ in subsets $S \subseteq \Gamma$ of size $m$ by letting $u_i \in S$ iff $s_i = 0$ and $v_i \in S$ iff $s_i = 1$. Different sequences of length $m$ then correspond to different orbits of $m$-subsets in $\Gamma$. 
  This contradicts the assumption that $\Gamma$ is sub-exponential.
\end{proof}

We now describe the automorphism groups of sub-exponential structures where the transitive constituents $G_1,\dots,G_k$
for all orbits $U_1,\dots,U_k$ of $G$
are isomorphic to $\Aut(({\mathbb Q};<))$: the following theorem shows that in this case $G$ is precisely what Cameron~\cite{Oligo} calls the \emph{intransitive action} of the direct product $G_1 \times \dots \times G_k$ on $U_1 \cup \dots \cup U_k$. 
We will say that a permutation group $G \subseteq \Sym(X)$ is \emph{transitive on a subset $Y \subseteq X$} if for all $x, y \in Y$, there exists a permutation $\alpha \in G$ such that $\alpha(x) = y$. Otherwise, we say that $G$ is \emph{intransitive on $Y$}.

\begin{theorem}\label{thm:no-interaction}
Let $\Gamma$ be a sub-exponential structure with automorphism group $G$ and
orbits $U_1, \dots, U_k$,
and for $1\leq i \leq k$ let $G_i$ be the transitive constituent
of $G$ on $U_i$.
Assume that $G_i$ is isomorphic to $\Aut(({\mathbb Q};<))$ 
for all $1 \leq i \leq k$.
Then, $\alpha \in \Sym(D)$ is an automorphism of $\Gamma$ if and only if
$\alpha|_{U_i} \in G_i$ for each $1 \leq i \leq k$.
 \end{theorem}

\begin{proof}
 Since $G_i$ is isomorphic to $\Aut(({\mathbb Q};<))$ for each $i$,
 there is by Theorem~\ref{thm:ryll} a first-order definable dense linear order
 $<_i$ on each orbit $U_i$.
 We will use $<$ to collectively denote these orders, and rely on the
context
 to determine which orbit it applies to.

 The result is trivial for $k = 1$, so assume that $k > 1$.
 Let $G'$ be the intransitive action of $G_1 \times \dots
\times G_k$ on $U_1 \cup \dots \cup U_k$.
 It suffices to show that $G$ is dense in $G'$; since $G$ is closed, they
must then be equal.
 Assume for the sake of contradiction that $G$ is not dense in $G'$.
 Then, there are finite sequences
 $({\bar u},u')$ and $({\bar v},v')$ in $\Gamma$,
 and an automorphism $\alpha \in G'$ such that
 $\alpha({\bar u},u') = ({\bar v},v')$, but
 $\gamma({\bar u},u') \neq ({\bar v},v')$ for all $\gamma \in G$.
 Furthermore we can choose $({\bar u},u')$ and $({\bar v},v')$
 so that $\beta({\bar u}) = {\bar v}$ for some $\beta \in G$.
 By applying $\beta^{-1}$ to ${\bar v}$ and $\alpha$, 
 we may then assume that
 ${\bar v} = {\bar u}$ and $\alpha({\bar u}) = {\bar u}$.
 By a \emph{${\bar u}$-interval}, we will mean an inclusion-maximal convex
 subset of $U_i \setminus {\bar u}$, for some $i$,
 where convexity is evaluated with respect to $U_i$.
 Now $u'$ lies in some ${\bar u}$-interval $I$,
 and since $v' = \alpha(u')$, we have $v' \in I$ as well.
 By assumption we have $\gamma(u') \neq v'$ for all $\gamma \in \Aut((\Gamma,{\bar u}))$,
 so $\Aut((\Gamma,{\bar u}))$ is intransitive on $I$.
 On the other hand, if we let ${\bar w} := {\bar u} \cap U$, where $U$ is the orbit of $\Gamma$ containing $I$, then
 $\Aut((\Gamma, {\bar w}))$ is clearly transitive on each ${\bar w}$-interval contained in $U$:
 each ${\bar w}$-interval contained in $U$ is an orbit, and its corresponding transitive constituent of $\Aut((\Gamma, {\bar w}))$ is isomorphic to $\Aut(({\mathbb Q}; <))$.
 In particular, $\Aut((\Gamma,{\bar w}))$ is transitive on $I$. It therefore follows that we can find a subsequence ${\bar a}$ of ${\bar u}$ containing ${\bar w}$, and an element $b \not\in U$ such that $\Aut((\Gamma, {\bar a}))$ is transitive on $I$ but $\Aut((\Gamma, {\bar a}, b))$ is not.

 By Lemma~\ref{lem:convexorbits}, we have that every orbit of $(\Gamma, {\bar a}, b)$ contained in $I$ is convex.
 Since $\Aut((\Gamma, {\bar a}, b))$ is intransitive on $I$, it follows
 that there is an initial segment 
 $I(b) \subsetneq I$
 definable in $(\Gamma, {\bar a}, b)$.
 Let $x \in I(b)$, $y \in I \setminus I(b)$, and pick an automorphism
 $\alpha \in \Aut((\Gamma, {\bar a}))$ such that $\alpha(x) = y$.
 Now $\alpha(I(b)) \supseteq I(b)$ and $\alpha(I(b))$ is definable in
 $(\Gamma, {\bar a}, \alpha(b))$.
 Let $b_1 = b$ and for $i \geq 1$, let $b_{i+1} = \alpha(b_i)$.
 By repeating this procedure, we get an increasing sequence of sets
 $I(b_1) \subsetneq I(b_2) \subsetneq \dots \subsetneq I(b_m)$ which are all
 definable in $(\Gamma, {\bar a}, {\bar b})$, 
 where ${\bar b} = (b_1, \dots, b_m)$.
 For $i \geq 1$, pick $c_i \in I(b_{i+1}) \setminus I(b_i)$,
 so that each of the elements $c_i$ lies in a
 different orbit of $(\Gamma, {\bar a}, {\bar b})$.
 We now encode a set $S \subseteq \{1,\dots,m\}$ as the subset $T$
 of $\Gamma$ consisting of the elements in
 ${\bar a}$, ${\bar b}$, and $\{ c_i \mid i \in S\}$.
 Let $S' \subseteq \{1,\dots,m\}$ be another set with encoding $T'$, and
 let $\alpha \in G$ be an automorphism such that $\alpha(T) = T'$.
 Then $\alpha$ has to fix ${\bar a}$ and ${\bar b}$,
 i.e., $\alpha$ must be an automorphism of $(\Gamma, {\bar a}, {\bar b})$.
 Therefore $\alpha$ cannot map $c_i$ to $c_j$ for $i \neq j$,
 so $S$ must be equal to $S'$.
 The set $T$ has size at most $2m+|{\bar a}|$,
 and since we can fix the size of ${\bar a}$,
 we conclude that the number of distinct orbits of $m$-subsets of $\Gamma$
 is at least $\Omega(2^{m/2})$.
 This contradicts the sub-exponentiality of $\Gamma$,
 so $G$ must be dense in $G'$, and the result follows.
\end{proof}

Together with the remarks of Section~\ref{sect:mc-cores},
the following implies Theorem~\ref{thm:alg-cond}.

\begin{theorem}
Every sub-exponential model-complete core $\Gamma$ with totally symmetric polymorphisms of all arities and a finite
relational signature $\tau$ has an efficient sampling algorithm.
\end{theorem}


\begin{proof}
For a given $n$, we have to compute a $\tau$-structure $B$ in polynomial time in $n$ such that for all structures $A$ of size $n$ 
we have $A \rightarrow \Gamma$ if and only if $A \rightarrow B$.
Let $U_1,\dots,U_k$ be the orbits of $\Gamma$. 
Let $\Gamma'$ be the expansion of $\Gamma$ by all primitive positive definable relations, and let $\Delta'_1,\dots,\Delta'_k$ 
be the structures induced in $\Gamma'$ by the orbits $U_1,\dots,U_k$ of $\Gamma$. By Proposition~\ref{prop:trans-mc-cores}, $\Delta'_i$ is
a transitive model-complete core, for all $i \leq k$. 
Since $\Delta'_i$ also has totally symmetric polymorphisms of all arities (obtained as the restrictions of the totally symmetric polymorphisms of $\Gamma$ to $U_i$), we can apply
Theorem~\ref{thm:transitive}, and conclude that each of the structures $\Delta'_i$ either has size 1, or is isomorphic to a structure which is first-order interdefinable with $({\mathbb Q};<)$.
We will now prove that an efficient sampling procedure exists in the case when each of the structures $\Delta'_i$ is isomorphic to a structure which is first-order interdefinable with $({\mathbb Q};<)$.
The proof can easily be modified to handle the case when some of the structures $\Delta'_i$ have size 1.

Let $m$ be the maximal arity of $\tau$, and
let $\sigma$ be the signature that contains a relation symbol for each at most $m$-ary primitive positive definable relation in $\tau$. For all $i \leq k$, let $\Delta_i$ be the structure obtained
from $\Delta_i'$ by removing all relations except the relations
for the symbols from $\sigma$. Then $\Delta_i$ has a finite signature, and we conclude by Theorem~\ref{thm:transitive} that $\Delta_i$ has an efficient sampling procedure.

Let $B_i$ be the $\sigma$-structure produced by this sampler for $\Delta_i$ on input $n$; 
we can assume that it has exactly $n$ vertices $u^i_1,\dots,u^i_n$. 
The output of our algorithm will be the $\tau$-reduct of a $\sigma$-structure $B$ with vertex set $\{u^i_j \mid 1 \leq j \leq n, 1 \leq i \leq k\}$.
Since the orbits are primitive positive definable in $\Gamma$, there is a unary relation symbol $R(U)$ in $\sigma$ for each orbit $U$ of $\Gamma$.
The structure $B$ will be such that $u^i_j \in R(U_i)$, for all $i \leq k$ and $j \leq n$.
Note that by Theorem~\ref{thm:no-interaction}, all such structures $B$ will be isomorphic.

For $R \in \tau$, we now add tuples to the relation $R^B$ as follows. 
We first fix a partition $P$ of the arguments of $R$ into at most $k$ parts. 
Let $p^i_1,\dots,p^i_{l(i)}$ be the arguments of $R$ of the $i$-th part.
For each $i$ the relation 
\begin{align*} \big \{(t[p^i_1],\dots,t[p^i_{l(i)}]) \mid t \in R \text{ and } t[p^j_l] \in U_j \text{ for all } j \leq k,l \leq l(j)  \big \}
\end{align*}
is primitive positive definable in $\Gamma$, and there is a relation symbol $R(P,i)$ for this relation in $\sigma$. 
It is clear that given $P$ and the relations $R(P,i)$ of $B_i$ we can efficiently compute the relation 
$$R(P) = \big \{t \mid (t[p^i_1],\dots,t[p^i_{l(i)}]) \in R(P,i)  \text{ for all } i \leq k  \big \} \; .$$
It follows from Theorem~\ref{thm:no-interaction} that $R$ equals the union of $R(P)$ over all partitions $P$ of the arguments of $R$ -- and since there is a constant number of such partitions, $R$ can be computed in polynomial time. 
\end{proof}

We are left with the task to prove Theorem~\ref{thm:transitive}.

\subsection{The Transitive Case is Primitive}
\label{ssect:finest}

We now proceed to show that when $\Gamma$ is a transitive sub-exponential
model-complete core with totally symmetric polymorphisms of all arities,
then it is also primitive.
Combining this result with the following two theorems shows that
$\Gamma$ is isomorphic to a structure with a first-order definition in $({\mathbb Q}; <)$.
This allows us to finish the proof of Theorem~\ref{thm:transitive} in the next section. 

\begin{theorem}[cf.~\cite{MacphersonOrbits}]
\label{thm:macpherson}
Let $G$ be a primitive but not highly set-transitive permutation group on an infinite set $X$.
If $c$ is a real number with $1 < c < 2^{1/5}$, then 
$G$ has more than $c^n$ orbits of $n$-subsets of $X$, for all sufficiently large $n$.
\end{theorem}

\begin{theorem}[Cameron~\cite{Cameron5}]
\label{thm:cameron}
  A permutation group $G$ on an infinite set is highly set-transitive iff it is isomorphic to the automorphism group of a structure with a first-order definition in $({\mathbb Q}; <)$.
\end{theorem}

We will need two more lemmas.
Let $R$ and $S$ be two binary relations.
An \emph{alternating closed walk on $R$ and $S$ of length $2n$} is a sequence of elements 
$(x_0,x_1,\dots,x_{2n})$, with $x_{2n} = x_0$, and such that
$(x_{2i},x_{2i+1}) \in R$ and $(x_{2i+1},x_{2i+2}) \in S$, for $0 \leq i < n$.

\begin{lemma} \label{lem:aclwalk}
  Let $R$ and $S$ be two binary relations that are preserved by a
  totally symmetric function $f_n$ of arity $n \geq 1$. 
  If there is an alternating closed walk on $R$ and $S$ of length $2n$,
  then $R \cap S^{-1} \neq \emptyset$.
\end{lemma}

\begin{proof}
  Since $(x_{2i},x_{2i+1}) \in R$ for $0 \leq i < n$, we have
  $(y,z) \in R$, for $y = f_n(x_0,x_2,\dots,x_{2n-2})$ and
  $z = f_n(x_1,x_3,\dots,x_{2n-1})$.
  Similarly, since $(x_{2i+1},x_{2i+2}) \in S$ for $0 \leq i < n$, we have
  $(y',z') \in S$, for $y' = f_n(x_1,x_3,\dots,x_{2n-1})$ and
  $z' = f_n(x_2,x_4,\dots,x_{2n})$.
  Note that
  \[
  y = f_n(x_0,x_2, \dots, x_{2n-2}) =
  f_n(x_2,x_4, \dots, x_{2n}=x_0) = z',
  \]
  and that $z = y'$.
  Therefore, $(y,z) \in R$ and $(z,y) \in S$,
  hence $(y,z) \in R \cap S^{-1}$.
\end{proof}

An \emph{orbital} of $\Gamma$ is an orbit of $\Aut(\Gamma)$ acting component-wise on ordered pairs of elements of $\Gamma$.
Every structure always has the \emph{trivial orbital} $\{ (x,x) \mid x \in \Gamma \}$.

\begin{lemma} \label{lem:k1}
  Let $\Gamma$ be a model-complete core with finitely many orbitals and a totally symmetric polymorphism $f_n$ of arity $n$ for all $n \geq 1$.
  Let $X$ be an equivalence class of a first-order definable equivalence relation on $\Gamma$.
  If $\alpha^m(X) = X$ for some $\alpha \in \Aut(\Gamma)$ and $m > 0$,
  then $\alpha(X) = X$.
\end{lemma}

\begin{proof}
  By assumption,
  there exists a smallest integer $r \geq 1$ such that $\alpha^r(X) = X$.
  Let $x \in X$, and
  for $k \in \mathbb{Z}$, let $O(k)$ be the orbital of $\Gamma$ containing the tuple $(x,\alpha^k(x))$.
  Then, we have the inclusion $\{ (\alpha^{n}(x), \alpha^{n+k}(x)) \mid n \in \mathbb{Z} \} \subseteq O(k)$.
  
  Since $\Gamma$ has finitely many orbitals,
  we can find integers $0 < l < k$ such that $O(k) = O(l)$.
  In fact, we can do this while ensuring that $l \equiv r-1 \text{ (mod $r$)}$.
  Note that $(\alpha^{i}(x),\alpha^{i+l+1}(x)) \in O(l+1)$, and that
  $(\alpha^{i+1+l}(x),\alpha^{i+1}(x)) \in O(-l)$, for all $i$.
  In particular, the following sequence is an alternating closed walk on $O(l+1)$ and $O(-l) = O(l)^{-1}$ of length $2(k-l)$.
  \[
  (x,\alpha^{l+1}(x),\alpha^1(x),\alpha^{l+2}(x),\alpha^2(x), \dots, \alpha^{k-l-1}(x),\alpha^{k}(x),x)
  \]

  As $\Gamma$ is a model-complete core,
  each orbital is primitive positive definable in $\Gamma$, and hence 
  preserved by $f_n$ for each $n \geq 1$.
  From Lemma~\ref{lem:aclwalk}, it now follows that the orbitals
  $O(l+1)$ and $O(l)$ intersect, and therefore they must be equal.
  This implies that the tuples $(x,\alpha^{l+1}(x))$ and $(x,\alpha^{l}(x))$ 
  are in the same orbital, so there exists an automorphism $\beta$ of $\Gamma$
  which fixes $x$ and maps $\alpha^{l+1}(x)$ to $\alpha^{l}(x)$.
  Since $\beta$ fixes $x$, we have $\beta(X) = X$, and from the choice of
  $l$, we have $\alpha^{l+1}(X) = X$.
  It follows that $\alpha^{l}(x) = \beta(\alpha^{l+1}(x)) \in X$,
  so $\alpha^{l}(X) = X$, and hence $\alpha^{r-1}(X) = X$ as well.
  Due to the choice of $r$, this is only possible if $r = 1$,
  so we conclude that $\alpha$ preserves $X$.
\end{proof}

We are now ready to prove the main result of this section.

\begin{proposition} \label{prop:transisprim}
  Let $\Gamma$ be a transitive sub-exponential model-complete core with totally symmetric polymorphisms of all arities. Then $\Gamma$ has size 1, or it is infinite and primitive.
\end{proposition}

\begin{proof}
  Assume that $\Gamma$ has size at least 2.
  Let $E_0$ and $E_1$ denote the congruence of $\Aut(\Gamma)$ with equivalence classes of size 1, and the congruence with a single equivalence class, respectively.
  Let $E$ be an inclusion-maximal congruence
  from the set of all congruences different from $E_1$.
  Existence of $E$ follows from $\Gamma$ having finitely many first-order definable equivalence relations, and the existence of $E_0$. 
  We want to show that $E$ must in fact be $E_0$
  from which it follows that $\Aut(\Gamma)$ is primitive.

  Let $D = D(\Gamma)$.
  By $D/{E}$ we will denote the set of equivalence classes of $E$.
  For $x \in D$, let $x[E]$ denote the equivalence class of $E$ containing $x$, and
  for $\alpha \in \Aut(\Gamma)$, let $\alpha/{E}$ denote the function on $D/{E}$ which maps $x[E]$ to $\alpha(x)[E]$ for each equivalence class $x[E]$ of $E$.
  (It follows from $E$ being a congruence that $\alpha/{E}$ is well-defined.)
  We then have that $\Aut(\Gamma)/{E} := \{ \alpha/{E} \mid \alpha \in \Aut(\Gamma) \}$
  is a permutation group on $D/{E}$.

  Let $H = \Aut(\Gamma)/{E}$.
  For each equivalence class $X_i \in D/E$, choose $x_i \in D$ such that $x_i[E] = X_i$.
  Let $\{X_1,\dots,X_n\}$ and $\{Y_1,\dots,Y_n\}$ be two $n$-subsets of $D/E$
  such that $\alpha(\{x_1,\dots,x_n\}) = \{y_1,\dots,y_n\}$ for some $\alpha \in \Aut(\Gamma)$.
  Then, $\{\alpha/E(X_1),\dots,\alpha/E(X_n)\} = \{Y_1,\dots,Y_n\}$, so the number of orbits of $n$-subsets of $H$ is sub-exponential whenever $\Gamma$ is sub-exponential.

  Assume that $H$ is finite.
  Pick two distinct equivalence classes $X$ and $Y$ of $E$.
  Since $\Gamma$ is transitive, we can find $\alpha \in \Aut(\Gamma)$ such that $\alpha(X) = Y$.
  But if $H$ is finite, then $\alpha^m(X) = X$ for some $m > 0$,
  so $\alpha(X) = X$ by Lemma~\ref{lem:k1}, a contradiction.
  Therefore $H$ must be infinite.
  Congruences of $H$ are in one-to-one correspondence with the congruences of $\Aut(\Gamma)$ containing $E$.
  As the latter are precisely $E$ and $E_1$, it follows that $H$ is primitive.

  It now follows from Theorem~\ref{thm:macpherson} that $H$ is highly set-transitive, and so from Theorem~\ref{thm:cameron} that $H$ is isomorphic to the automorphism group $H'$ of a structure with a first-order definition in $({\mathbb Q}; <)$.
  The group $H'$ either has one or two non-trivial orbitals.
  If it only has one non-trivial orbital, then so does $H$,
  hence for any two distinct equivalence classes $X$ and $Y$ of $E$, there is an automorphism $\alpha$ of $\Gamma$ such that $\alpha(X)=Y$ and $\alpha(Y) = X$.
  Again by Lemma~\ref{lem:k1}, it follows that $\alpha(X) = X$, a contradiction.
  So $H'$ has two non-trivial orbitals, one of which is the order $<$ on ${\mathbb Q}$.
  Via the isomorphism, $H$ thus has a non-trivial orbital $<$ which is a linear order on the equivalence classes of $E$.
  Let $R$ be a binary relation on $\Gamma$ defined by $R(x,y)$ iff $x[E] < y[E]$.
  For $\alpha \in \Aut(\Gamma)$, we have that $\alpha/{E}$ preserves $<$, so $\alpha$ preserves $R$.
  Assume now that the equivalence classes of $E$ have size greater than 1.
  We can then encode a sequence in $\{0,1\}^n$ as a set 
$\{x_1,y_1,\dots,x_n,y_n\} \subseteq \Gamma$:
  choose $x_i, y_i$ so that $R(y_i,x_{i+1})$ for $1 \leq i < n$,
  encode a value 0 in position $i$ by enforcing $E(x_i,y_i)$, $x_i \neq y_i$, and encode a value 1 by enforcing $R(x_i,y_i)$.
  The relations $E$, $\neq$, and $R$ are all preserved by $\Aut(\Gamma)$, so
  if two $2n$-subsets encodes distinct sequences, then they
  must be contained in distinct orbits of $2n$-subsets.
  Hence, the number of orbits of $2n$-subsets is greater than or equal to $2^n$, 
  which contradicts $\Gamma$ being sub-exponential.
  So the equivalence classes of $E$ are of size 1, i.e., $E = E_0$, and $\Gamma$ is primitive.
\end{proof}

\subsection{The Primitive Case}
\label{ssect:primitive}
Assume that $\Gamma$ is a transitive sub-exponential model-complete core with totally symmetric polymorphisms of all arities.
If $\Gamma$ is of size greater than 1, then it is infinite and primitive (Proposition~\ref{prop:transisprim}),
so it is highly set-transitive (Theorem~\ref{thm:macpherson})
and hence isomorphic to a structure $\Gamma'$ with a first-order definition in
$({\mathbb Q};<)$ (Theorem~\ref{thm:cameron}).
To prove the first part of Theorem~\ref{thm:transitive},
we have to show that $\Gamma'$ is in fact first-order
interdefinable with $({\mathbb Q};<)$. 
If the binary relation $<$ is an orbital of $\Gamma'$, then
$<$ is first-order definable in $\Gamma'$ by Theorem~\ref{thm:ryll}, and interdefinability follows. 
Otherwise, the smallest orbital of $\Gamma'$ that contains $<$ also contains a pair $(x,y)$ such that $x > y$.
It follows that $\Gamma'$ has an automorphism $\alpha$ such that $\alpha(x) = y$ and $\alpha(y) = x$, i.e., $\alpha^2(x) = x$. Since $\Gamma'$ has totally symmetric polymorphisms of all arities, we can apply Lemma~\ref{lem:k1} to deduce $\alpha(x) = x$, a contradiction.

To prove the second part of Theorem~\ref{thm:transitive},
we show that any structure with a finite relational
signature and a first-order definition in $({\mathbb Q};<)$ 
has an efficient sampling algorithm.
A relational structure $\Gamma$ is \emph{homogeneous} (or \emph{ultrahomogeneous}) if every isomorphism between finite induced substructures of $\Gamma$ can be extended to an automorphism of $\Gamma$.
A structure $\Gamma$ with signature $\tau$ \emph{admits quantifier-elimination} if every first-order $\tau$-formula is over $\Gamma$ equivalent to a quantifier-free $\tau$-formula.
An $\omega$-categorical structure admits quantifier-elimination iff it is homogeneous~\cite{Oligo}.
Let $[n] := \{1,\dots,n\}$.

\begin{theorem}\label{thm:primitive}
Every structure $\Gamma$ with a finite relational signature $\tau$ and 
a first-order definition in $({\mathbb Q};<)$ has an efficient sampling algorithm. On input $n$, the output of the algorithm is a structure
of size $n$.
\end{theorem}
\begin{proof}
The structure $({\mathbb Q};<)$ is homogeneous, 
and hence admits quantifier-elimination.
Let $\phi_1,\dots,\phi_k$ be first-order definitions of the relations $R^\Gamma_1,\dots,R^\Gamma_k$ of $\Gamma$, 
written in quantifier-free conjunctive normal form.
For a given $n$, we compute a finite $\tau$-structure $B$ with domain $[n]$ as follows, where $[n]$ is viewed as a subset of $D(\Gamma) = {\mathbb Q}$.
For each $i$, let $R^B_i$ be the  $m_i$-ary relation over the domain $[n]$ that contains
all $m_i$-tuples that satisfy $\phi_i$ (where $m_i$ is the number of free variables of $\phi_i$). 
Since $\phi_i$ is of constant size in $n$, the relation $R^B_i$ can be computed in time 
$O(n^{m_i})$. 
The resulting structure $B = ([n]; R^B_1,\dots,R^B_k)$ 
is clearly a substructure of $\Gamma$. Moreover, if $A$ is a finite $\tau$-structure
with $n$ elements, and $s: A \rightarrow {\mathbb Q}$ is a homomorphism from $A$ to $\Gamma$,
then by high set-transitivity of $\Gamma$ there is an automorphism $\alpha$ of $\Gamma$ such that
$\alpha(s(A)) \subseteq [n]$. Hence, $A$ homomorphically maps to $\Gamma$ if and only if $A$ homomorphically
maps to $B$.
\end{proof}
\section{Beyond Sub-Exponential Growth: Examples}
\label{sect:beyond}

In this section we illustrate with a couple of examples the possibility
of extending our tractability result beyond sub-exponential structures.
All examples of $\omega$-categorical structures with totally symmetric polymorphisms of all arities that we are aware of 
have a first-order interpretation over the structure $({\mathbb Q};<)$.
Interpretations are a central concept from model theory, which we briefly recall in the following.
Let $\sigma$ and $\tau$ be relational signatures, $\Delta$ a $\sigma$-structure, and $\Gamma$ a $\tau$-structure.
A \emph{$d$-dimensional (first-order) interpretation $I$ of $\Gamma$ in $\Delta$} consists of (cf.~\cite{Hodges})

\begin{enumerate}
\item a $\sigma$-formula $\partial_I (x_1, \dots, x_{d})$; 
\item for each atomic $\tau$-formula $\phi(y_1, \dots, y_m)$,
     a $\sigma$-formula $\phi_I({\bar x_1}, \dots, {\bar x_m})$, where
     ${\bar x_1}, \dots, {\bar x_m}$ are disjoint $d$-tuples of distinct variables; and
\item a surjective map $f_I : \partial_I((D(\Delta))^d) \rightarrow \Gamma$; 
\end{enumerate}
such that for all atomic $\tau$-formulas $\phi$ and all ${\bar a_i} \in \partial_I(D(\Delta)^d)$, 
$\phi(f_I({\bar a_1}), \dots, f_I({\bar a_m}))$ holds in $\Gamma$ if and only if
$\phi_I({\bar a_1}, \dots, {\bar a_m})$ holds in $\Delta$.

Our interest in interpretations also stems from the following result.

\begin{lemma}
\label{lem:interpretq}
  Every structure $\Gamma$ with a finite relational signature $\tau$ and a
  $d$-dimensional interpretation in $({\mathbb Q};<)$ 
  has an efficient sampling algorithm.
  On input $n$, the output of the algorithm is a structure of size at most $(dn)^d$.
\end{lemma}

\begin{proof}
  Let $I$ be a $d$-dimensional interpretation of $\Gamma$ in $({\mathbb Q}; <)$.
  On input $n$, we compute a finite $\tau$-structure $B$ which is the induced substructure
  of $\Gamma$ on the domain $f_I(\partial_I([dn]^d))$.
  For an $m$-ary $R \in \tau$, we do the following.
  Let $\phi_I$ be the interpretation in $({\mathbb Q}; <)$ of the atomic formula $R(y_1,\dots,y_m)$,
  given in quantifier-free conjunctive normal form.
  We now evaluate $\phi_I$ on each sequence of $d$-tuples, 
  ${\bar a_1},\dots,{\bar a_m} \in \partial_I([dn]^d)$.
  The formula $\phi_I({\bar a_1},\dots,{\bar a_m})$ is true iff the tuple $(f_I({\bar a_1}), \dots, f_I({\bar a_m}))$ is in $R^B$.
  Since $\phi$ is of constant size in $n$, it follows that we can compute the relation $R^B$ in $O((dn)^{dm})$ time.
  The signature $\tau$ is finite, so there is a relation of highest arity, independent of $n$, 
  which provides the upper bound on the time complexity of the algorithm.
  
Next, let $A$ be a finite $\tau$-structure with $n$ elements, and assume that
$s: A \rightarrow \Gamma$ is a homomorphism.
This implies $A \rightarrow B$ as well:
the image of $A$ under $s$ has at most $n$ elements, $b_1, \dots, b_n \in \Gamma$.
 Let ${\bar a_1}, \dots, {\bar a_n}$ be tuples in $\partial_I({\mathbb Q}^d)$ such that $f_I({\bar a_i}) = b_i$, and
 let $g : s(A) \rightarrow \partial_I({\mathbb Q}^d$) be the function such that $g(b_i) = {\bar a_i}$ for all $i$.
Note that ${\bar a_1}, \dots, {\bar a_n}$ contains at most $dn$ distinct values of ${\mathbb Q}$.
By high set-transitivity of $({\mathbb Q}; <)$, it follows that there is an automorphism $\alpha$ of $({\mathbb Q}; <)$ such that $\alpha(a_{ij}) \in [dn]$
for each $i$ and $j$, where ${\bar a_i} = (a_{i1},\dots,a_{id})$.
Now for each relation symbol $R \in \tau$, we have $(x_1,\dots,x_m) \in R^\Gamma$ iff
$\phi_I(g(x_1),\dots,g(x_m))$ holds in $({\mathbb Q}; <)$ iff
$\phi_I(\alpha(g(x_1)),\dots,\alpha(g(x_m)))$ holds in $({\mathbb Q}; <)$,
and this in turn is true if and only if
$(f_I(\alpha(g(x_1))),\dots,f_I(\alpha(g(x_m)))) \in R^B$.
It follows that $f_I \circ \alpha \circ g \circ s$ is a homomorphism 
from $A$ to $B$.
\end{proof}

Below, we give a number of examples of structures with exponential growth that are interpretable in $({\mathbb Q};<)$, and have semi-lattice operations.

\begin{example}
Let `$<$', `$=$', and `$>$' denote the usual inequality and equality relations
on ${\mathbb Q}$.
Let $\Gamma_1$ be the relational structure over ${\mathbb Q}^2$ with 
binary relations $R_{\rho,\sigma}$ for $\rho, \sigma \in \{<,=,>\}$, where
$R_{\rho,\sigma}((x_1,y_1),(x_2,y_2))$ is defined by
$\rho(x_1,x_2) \wedge \sigma(y_1,y_2)$.
This structure is exponential.
To obtain a lower bound on the growth rate of $\Gamma_1$, 
pick an $n$-subset
$A = \{a_1, \dots, a_n\} \subseteq \mathbb{Q}^2$,
and assume that the projection of $A$ on the second component contains 
exactly $k$ distinct values.
For $1 \leq i \leq k$, let $A_i$ denote the number of elements that have
the $i$th largest second component value.
Then, $A_1 + \dots + A_k = n$ determines a \emph{composition of $n$},
i.e., an expression for $n$ as an ordered sum of positive integers.
Two sets $A, B \subseteq \mathbb{Q}^2$ which determine different compositions
must be in different orbits.
Thus, the number of orbits of $n$-subsets of $\mathbb{Q}^2$
is at least $2^{n-1}$, the number of compositions of $n$.

Hence, Theorem~\ref{thm:yyy} does not apply to $\Gamma_1$.
Instead, tractability can be inferred as follows.
The structure $\Gamma_1$ has a two-dimensional interpretation $I$ in $({\mathbb Q};<)$.
The formula $\partial_I$ is always true, 
$f_I$ is the identity on ${\mathbb Q}^2$, 
and the interpretations of $R$ and $S$ are as given above.
Furthermore, it is easy to verify that $\Gamma_1$ is invariant under the
semi-lattice operations given by (component-wise) $\min$ and $\max$.
Tractability of $\Csp(\Gamma_1)$ now follows from Theorem~\ref{thm:correctness} and Lemma~\ref{lem:interpretq}.

For any $d > 2$, this example can be generalized
to a structure with $3^d$ $d$-ary relations, and with a 
$d$-dimensional interpretation in $({\mathbb Q}; <)$.
The CSP of each such structure is polynomial-time solvable.
\qed \end{example}

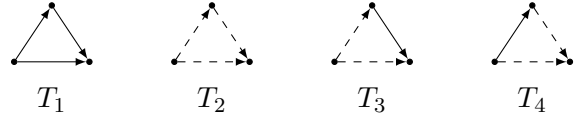
\begin{figure}[htbp]\label{fig:1}
\begin{center}
\begin{tikzpicture}
  \useasboundingbox (-1,-1) rectangle (1,1.5);
  \node () at (0,-0.5) {$T_1$};
  \node[vertex8] (c) at (0.5cm,0)  [] {};
  \node[vertex8] (b) at (0,0.75cm) [] {} edge[edge](c);
  \node[vertex8] (a) at (-0.5cm,0) [] {} edge[edge](c) edge[edge] (b);
\end{tikzpicture}
\begin{tikzpicture}
  \useasboundingbox (-1,-1) rectangle (1,1.5);
  \node () at (0,-0.5) {$T_2$};
  \node[vertex8] (c) at (0.5cm,0)  [] {};
  \node[vertex8] (b) at (0,0.75cm) [] {} edge[edged](c);
  \node[vertex8] (a) at (-0.5cm,0) [] {} edge[edged](c) edge[edged] (b);
\end{tikzpicture}
\begin{tikzpicture}
  \useasboundingbox (-1,-1) rectangle (1,1.5);
  \node () at (0,-0.5) {$T_3$};
  \node[vertex8] (c) at (0.5cm,0)  [] {};
  \node[vertex8] (b) at (0,0.75cm) [] {} edge[edge](c);
  \node[vertex8] (a) at (-0.5cm,0) [] {} edge[edged](c) edge[edged] (b);
\end{tikzpicture}
\begin{tikzpicture}
  \useasboundingbox (-1,-1) rectangle (1,1.5);
  \node () at (0,-0.5) {$T_4$};
  \node[vertex8] (c) at (0.5cm,0)  [] {};
  \node[vertex8] (b) at (0,0.75cm) [] {} edge[edged](c);
  \node[vertex8] (a) at (-0.5cm,0) [] {} edge[edged](c) edge[edge] (b);
\end{tikzpicture}
\end{center}
\caption{Four relational structures with signature $(R,S)$; the relations $R$ and $S$ are given by the solid and dashed arrows, respectively.}
\end{figure}

\begin{example}
The \emph{age} of a structure $\Gamma$ is defined as the class of all finite
structures isomorphic to a substructure of $\Gamma$.
Let $(R,S)$ be a signature with two binary relation symbols, and let
$T = \{T_1, T_2, T_3, T_4\}$ be the set of structures in Fig.~\ref{fig:1},
where the tuples of $R$ ($S$) are given by the solid (dashed) arrows.
Let ${\mathcal C}$ be the class of all finite structures with signature $(R,S)$
for which every three-element substructure is isomorphic to a structure in $T$.
It can be shown that ${\mathcal C}$ is an \emph{amalgamation class}
so that its \emph{Fra\"iss\'e limit} exists (cf.\ Theorem~6.1.2 in \cite{Hodges}).
This is the up to isomorphism unique countable homogeneous structure
with age ${\mathcal C}$.

The following describes a relational structure $\Gamma_2$ with age ${\mathcal C}$
which can be verified to be homogeneous.
It follows that $\Gamma_2$ is isomorphic to the Fra\"iss\'e limit of ${\mathcal C}$.
Let $\Gamma_2 = ({\mathbb Q}^2; R, S)$, where
$R((x_1,y_1),(x_2,y_2))$ is the relation $x_1 = x_2 \wedge y_1 < y_2$, and
$S((x_1,y_1),(x_2,y_2))$ is the relation $x_1 < x_2$.
The growth rate of $\Gamma_2$ can be bounded as in the previous example,
and here it turns out that the number of orbits of $n$-subsets is precisely
$2^{n-1}$.
The structure $\Gamma_2$ also has a two-dimensional interpretation in $({\mathbb Q};<)$
and semi-lattice polymorphisms given by (component-wise) $\min$ and $\max$,
so tractability follows once again from
Theorem~\ref{thm:correctness} and Lemma~\ref{lem:interpretq}.
\qed \end{example}


\begin{example}
  Let $\Gamma_3 := (U \cup V; M, <)$ be the following relational structure.
  The domain $U \cup V$ is the disjoint union of two copies of ${\mathbb Q}$.
  The binary relation $M$ defines a perfect matching between the elements
  of $U$ and the elements of $V$,
  and the binary relation $<$ defines a dense linear order on $U \cup V$
  such that $u < v$ for all $u \in U$ and $v \in V$,
  and for $v_1,v_2 \in V$, we have $v_1 < v_2$
  iff $u_1 < u_2$ for the elements $u_1, u_2 \in U$ with
  $(u_1,v_1), (u_2,v_2) \in M$.

  The structure $\Gamma_3$ is invariant under the semi-lattice operations given by $\min$ and $\max$ defined with respect to the order $<$ on $U \cup V$.
  It has two orbits and $\Aut(\Gamma_3)$ is isomorphic (as an abstract group) to $\Aut(({\mathbb Q}; <))$.
  By Theorem~\ref{thm:no-interaction}, this implies that $\Gamma_3$ does not have sub-exponential growth.
  But $\Gamma_3$ has a 2-dimensional interpretation $I$ in $({\mathbb Q}; <)$, so $\Csp(\Gamma_3)$ is polynomial-time solvable:
  let $\partial_I(x,y)$ be the formula $x \neq y$, 
  and let $f_I(x,y)$ be the copy of $x$ in $U$ if $x < y$ and the copy
  of $x$ in $V$ if $x > y$.
  The matching $M$ on $(x_1,y_1)$ and $(x_2,y_2)$ is interpreted by the formula
  $x_1 = x_2 \wedge x_1 < y_1 \wedge x_2 > y_2$
  and the order $<$ on $(x_1,y_1)$ and $(x_2,y_2)$ is interpreted by the formula
  $(x_1 < y_1 \wedge x_2 > y_2) \vee 
   (x_1 < y_1 \wedge x_2 < y_2 \wedge x_1 < x_2) \vee
   (x_1 > y_1 \wedge x_2 > y_2 \wedge x_1 < x_2)$.
\qed \end{example}

Every finite structure $B$ has an interpretation $I$ in $({\mathbb Q}; <)$,
in fact, even in $({\mathbb Q};=)$.
Let $n = |D|$ and the dimension of the interpretation be $d = 2n$.
The formula $\partial_I(x_1,\dots,x_n,x'_1,\dots,x'_n)$ is true
if and only if for exactly one $i$ it holds that $x_i = x'_i$.
Equality is interpreted by the formula
\[
\phi_=(x_1,\dots,x_n,x'_1,\dots,x'_n,y_1,\dots,y_n,y'_1,\dots,y'_n)
= (\bigwedge_{i=1}^n (x_i = x'_i) \Leftrightarrow (y_i = y'_i)).
\]
It is now straightforward to write down first-order formulas $\phi_R$
that interpret the relations of $B$.
When $R$ is $k$-ary, then the formula $\phi_R$ is a
disjunction of conjunctions with $2nk$ variables 
$x_{1,1},\dots,x_{1,n},x'_{1,1},\dots,x'_{1,n},\dots,x_{k,1},\dots,x_{k,n},x'_{k,1},\dots,x'_{k,n}$.
For each tuple $(t_1,\dots,t_k) \in R$, the disjunction contains the
conjunct $\bigwedge_{j=1}^k (x_{j,t_j} = x'_{j,t_j})$.

The examples in this section suggest the following question.

\begin{question}\label{quest:interpretations}
Is it true that all $\omega$-categorical relational structures with totally symmetric polymorphisms of all arities have a first-order interpretation
over $({\mathbb Q}; <)$?
\end{question}

\section{Concluding Remarks}
\label{sect:disc}
In this article we prove that constraint satisfaction problems for templates $\Gamma$
where the number of orbits of $n$-subsets of $\Gamma$ grows sub-exponentially in $n$ can be solved in polynomial time
when $\Gamma$ has a semi-lattice polymorphism. In fact, we showed the stronger result which only requires
the existence of totally symmetric polymorphisms of all arities, instead of requiring the existence of a semi-lattice polymorphism.
This algorithmic result can be showed in two stages: 
\begin{enumerate}
\item In the first stage, we reduce CSP$(\Gamma)$, for structures $\Gamma$ with totally symmetric polymorphisms of all arities, to solving certain uniform
finite domain CSPs, and to the task to find an efficient sampling algorithm for $\Gamma$. 
\item In the second stage, we \emph{classify} sub-exponential structures $\Gamma$ with totally symmetric polymorphisms of all arities,
and use the classification to verify that there always exists an efficient sampling algorithm for $\Gamma$.
\end{enumerate}

The reduction presented in the first stage crucially relies on the fact that 
when $\Gamma$ has totally symmetric polymorphisms of all arities,
then for all finite induced substructures $S$ of $\Gamma$ the set structure of $S$ homomorphically maps to $\Gamma$ (Lemma~\ref{lem:totsymtohomo}). We want to remark that this connection has a converse when $\Gamma$ is $\omega$-categorical. 

\begin{lemma}
\label{lem:omegaequiv}
  Let $\Gamma$ be an $\omega$-categorical structure over a finite relational signature.
  Then
  $\Gamma$ has totally symmetric polymorphisms of all arities if and only if
  ${\mathcal P}(S) \rightarrow \Gamma$ for all finite substructures
  $S \subseteq \Gamma$.
\end{lemma}

\begin{proof}
  The forward direction was proved in Lemma~\ref{lem:totsymtohomo}.
  The remaining direction can be proved using a common technique for
  constructing homomorphisms to $\omega$-categorical structures.
  Given an arbitrary positive integer $n$,
  we want to produce an $n$-ary totally symmetric polymorphism $f$ of $\Gamma$.
  The idea of the proof is as follows:
  let $l_1, l_2, \dots$ be an enumeration of the elements of $D = D(\Gamma)$,
  and let $L_k = \{l_1, \dots, l_k\}$.
  For each $k \geq 1$, let $F_k$ be the set of homomorphisms from ${\mathcal P}(\Gamma[L_k])$ to $\Gamma$.
  Introduce an equivalence relation $\sim$ on $F_k$ by defining $f \sim g$ iff $f = \alpha \circ g$ for some automorphism $\alpha$ of $\Gamma$.
  Let $\tilde{F}_k$ denote the set of equivalence classes of $F_k$ under $\sim$.

  Arrange the elements of $\bigcup_{k \geq 1} \tilde{F}_k$ into a
  forest containing at least one infinite tree:
  each $f_1 \in F_1$ is defined to be the root of a separate tree, and
  for each $k > 1$, and $f_k \in F_k$, define the parent of 
  $\tilde{f}_k \in \tilde{F}_k$ to be the
  equivalence class containing the restriction of $f_k$ to the non-empty
  subsets of $L_{k-1}$.
  This definition is independent of the choice of representative in $\tilde{f}_k$,
  so each equivalence class in $\tilde{F}_k$, $k > 1$, has precisely one parent.
  Since $\Gamma$ is $\omega$-categorical, it follows that there are finitely
  many equivalence classes for a fixed $k$.
  Hence, there are finitely many trees and each tree is finitely branching in each node.
  By assumption, $\tilde{F}_k$ is non-empty for each $k \geq 1$, 
  so some tree has unbounded height.
  Now, K\"onig's tree lemma implies the existence of an infinite path $\tilde{f}_1, \tilde{f}_2, \dots$ in some tree.
  Assume that there are representatives $f_1 \in \tilde{f}_1, f_2 \in \tilde{f}_2, \dots, f_{k} \in \tilde{f}_{k}$, such that $f_{k-1}$ is the restriction of $f_k$ to the non-empty subsets of $L_{k-1}$.
  We show that this path can be extended indefinitely:
  choose $g_{k+1} \in \tilde{f}_{k+1}$ arbitrarily and let $g_{k}$ be its
  restriction to the non-empty subsets of $L_{k}$.
  Then, there exists an automorphism $\alpha$ such that $f_{k} = \alpha \circ g_{k}$.
  It follows that $f_{k}$ is the restriction of $\alpha \circ g_{k+1}$
  to the non-empty subsets of $L_{k}$, hence we can define $f_{k+1} := \alpha \circ g_{k+1}$.

  Now, for any $n$-tuple $(x_1,\dots,x_n)$ over $D$, define
  $f(x_1,\dots,x_n) = f_m(\{x_1,\dots,x_n\})$, where $m$ is an
  any integer such that $\{x_1,\dots,x_n\} \subseteq L_m$.
  By the construction of the sequence $f_1, f_2, \dots$,
  the function $f$ is 
  a well-defined totally symmetric $n$-ary function on $D$.
  To verify that $f$ is a polymorphism of $\Gamma$,
  let $R^\Gamma$ be an $r$-ary relation, and $t_1, \dots, t_n \in R^\Gamma$.
  Let $U_i = \{t_1[i],\dots,t_n[i]\}$, for $1 \leq i \leq r$.
  Assume without loss of generality that $m$ has been chosen large enough
  so that $\bigcup_{i=1}^r U_i \subseteq L_m$.
  Then $f(t_1[i],\dots,t_n[i]) = f_m(U_i)$ for all $i$.
  By definition of the set structure, we have
  $(U_1,\dots,U_r) \in R^{{\mathcal P}(\Gamma[L_m])}$.
  Since $f_m$ is a homomorphism, 
  we conclude that $(f_m(U_1),\dots,f_m(U_r)) \in R^{\Gamma}$, so
  $f$ is indeed a polymorphism.
\end{proof}

Because of the general applicability of the algorithmic approach in
Section~\ref{sect:alg}, and because
of the fact that the known structures of exponential growth that have
totally symmetric polymorphisms seem to be well-behaved (see Section~\ref{sect:beyond}),
 we make the following conjecture.

\begin{conjecture}\label{conj:brave}
Let $\Gamma$ be an $\omega$-categorical structure with finite relational
signature and totally symmetric polymorphisms of all arities.
Then $\Csp(\Gamma)$ can be solved in polynomial time.
\end{conjecture}


\end{document}